\documentclass[12pt, onecolumn]{IEEEtran}
\usepackage{epsf}
\usepackage{graphicx}
\usepackage{amsmath} \usepackage{amssymb}

\newtheorem{theorem}{Theorem}[section]

\newtheorem{lemma}{Lemma}[section]

\newtheorem{remark}{Remark}[section] 
\long\def\symbolfootnote[#1]#2{\begingroup%
\def\thefootnote{\fnsymbol{footnote}}\footnote[#1]{#2}\endgroup}

  \def\A{{\cal A}}
\def\BB{{\cal B}} \def\DD{{\cal D}}  
 \def\II{{\cal I}} \def\MM{{\cal M}} 
  \def\C{\mathbb C}

  \def\C{\Sigma}

\def\dref#1{(\ref{#1})}

\def\be{\begin{equation}} \def\ee{\end{equation}}
\def\ba{\begin{array}} \def\ea{\end{array}} \def\bna{\begin{eqnarray}}
\def\ena{\end{eqnarray}}

 \def\NN{{\cal N}} \def\MM{{\cal M}}

\def\C{\mathbb C}

\def\DD{{\cal D}}
\def\SS{{\cal S}}
    \def\JJ{{\cal J}}
\def\DD{{\cal D}} \def\EE{{\cal E}} \def\GG{{\cal G}}
\def\TT{{\cal T}}
\def\II{{\cal I}}

 \def\bna{\begin{eqnarray}}
\def\ena{\end{eqnarray}} \def\dref#1{(\ref{#1})}
\begin{document}
\title{ Omnidirectional Relay in Wireless Networks}

\author{\authorblockN{Liang-Liang Xie} \\
\authorblockA{\small Department of Electrical and Computer Engineering\\
University of Waterloo, Waterloo, ON, Canada N2L 3G1 \\
Email: llxie@ece.uwaterloo.ca} }

\maketitle

\begin{abstract}

For wireless networks with multiple sources, an omnidirectional
relay scheme is developed, where each node can simultaneously relay different
messages in different directions. This is accomplished by the
decode-and-forward relay strategy, with each relay binning the
multiple messages to be transmitted, in the same spirit of network
coding. Specially for the all-source all-cast problem, where each
node is an independent source to be transmitted to all the other
nodes, this scheme completely eliminates interference in the whole
network, and the signal transmitted by any node can be used by any
other node. For networks with some kind of symmetry, assuming no beamforming is to be performed, this omnidirectional relay scheme is capable of achieving the maximum achievable rate.

\end{abstract}

\section{Introduction}

In wireless networking, relay is a way of expanding communication range or increasing communication rate, with the help of other nodes. As such, more nodes are involved and more signals will be transmitted. It is therefore important to design and coordinate these signals to maximize the cooperation and minimize the interference.
Between the two fundamental relay strategies proposed in
\cite{covelg79}, especially, the decode-and-forward strategy enables the destination node to fully enjoy the
transmitted power of both the source node and the relay
node. This is still realizable when multiple relays are introduced
to help the destination \cite{xiekum04,xiekum05,kragasgup05}, and interference can be
completely eliminated for arbitrarily large networks.

However, the situation is much more complicated when there are
multiple sources in the network \cite{xiekum07}. Unlike the case
of a single source where all nodes are essentially transmitting
the same information, multiple sources seem inevitably result in
interference. Nevertheless, studies of the two-way relay channel
\cite{ranwit06,xie07} have indicated the possibility of no interference
even if there are more than one sources.

In this paper, we develop an omnidirectional relay scheme for
wireless networks with multiple sources, where, each node can
simultaneously relay different messages in different directions.
This is accomplished by binning multiple messages at each relay,
as a generalization of the scheme proposed in \cite{xie07}, in the same spirit of network coding \cite{ahlcailiyeu00}.

The basic idea of network coding \cite{ahlcailiyeu00} can be explained with the following example. Suppose that node $A$ wants to send out two bits of information $b_1$ and $b_2$, with $b_1$ to node $B$, and $b_2$ to node $C$. However, if node $B$ already knows $b_2$ and node $C$ already knows $b_1$, then this can be accomplished by just sending out one bit $b_1\oplus b_2$ to both node $B$ and node $C$, since node $B$ can recover $b_1$ by computing $b_2\oplus (b_1\oplus b_2)=b_1$, and node $C$ can recover $b_2$ by computing $b_1\oplus (b_1\oplus b_2)=b_2$.

This scheme can be generalized with the technique of binning \cite{xie07}. Consider the problem that node $A$ wants to send out two messages $w_1$ and $w_2$, with $w_1$ to node $B$, and $w_2$ to node $C$, where, $w_1$ can take $M_1$ different values and $w_2$ can take $M_2$ different values, and possibly, $M_1\neq M_2$. Similarly, assume that node $B$ already knows the true value of $w_2$, and node $C$ already knows the true value of $w_1$. Instead of sending out both the messages $(w_1,w_2)$, which can be any of the $M_1M_2$ different vectors, node $A$ can throw these vectors into $M$ bins, with $M= \max\{M_1,M_2\}$, and send out the index of the bin that contains the true vector. In this way, node $A$ only needs to send out a message with $M$ different values. It can be easily checked that when $M\geq \max\{M_1,M_2\}$, it is possible to bin the $M_1M_2$ different vectors of $(w_1,w_2)$ in such a way that in each bin, no two vectors contain the same $w_1$ or the same $w_2$. Therefore, knowing the true value of $w_2$, and the bin that contains the true value of $(w_1,w_2)$, node $B$ can uniquely determine the true value of $w_1$. Similarly, node $C$ can uniquely determine the true value of $w_2$.

The above binning scheme can be easily generalized to send any number of messages. In the context of wireless relay networks, node $A$ can be a relay that wants to forward different messages to different nodes. With the binning technique, node $A$ only needs to send one signal representing the bin index, from which, different receivers can pick up different messages based on their different {\it a priori}\, knowledge of the messages. Furthermore, with this binning scheme, it is also shown in \cite{xie07} that every receiver can fully exploit all the signal power, as if node $A$ is only sending those messages unknown to it.

Node $A$ can also use other ways to relay multiple messages, e.g., by superposition coding. It can first encode each message individually by a signal, and then superpose them together into a layered signal to transmit. Upon receiving this layered signal, each receiver can pick out the layers that correspond to the unknown messages, by deleting the layers that correspond to the messages already known. Compared to the binning scheme, an obvious drawback of this superposition scheme is that the total transmit power of node $A$ has to be clearly divided among the messages, and each receiver can only exploit the part that is used for its unknown messages. However, this way of clearly layering different messages makes it easier to establish cooperation between different transmitters. For example, to send the same message to a common receiver, beamforming or coherent transmission can be established between two transmitters so that the received power can be boosted. On the other hand, this is not so easy to realize with the binning scheme unless the two transmitters are sending exactly the same set of messages.

Using superposition coding to establish coherent transmission was originally proposed in \cite{covelg79} for the relay channel. It was later extended to the case with multiple relays \cite{xiekum04,xiekum05,kragasgup05}, and to the two-way relay channel \cite{ranwit06}. It can also be applied to a general framework with multiple sources, relays and destinations \cite{xiekum07}. However, the corresponding achievable rate regions become extremely messy for general networks, when there are too many layers of signals to consider. In this paper, we only consider the binning technique in the omnidirectional relay scheme.

As a special application which may be the best to demonstrate the benefit of this binning scheme, we consider the all-source
all-cast problem, where each node is an independent source, to be
sent to all the other nodes. We will show that for such problems,
it is possible to completely eliminate interference in the
network, and each node will enjoy the power transmitted by all the other nodes.

The remainder of the paper is organized as the following. In Section \ref{scheme}, we introduce a general framework of omnidirectional relay with arbitrary source-destination distributions in mind. Starting from Section \ref{asac}, we will focus on the all-source all-cast problem. First, a special version of the omnidirectional relay scheme is developed in Section \ref{asac} for the all-source all-cast problem. Then a key technical lemma is presented in Section \ref{lem}, before we prove some achievability results in Section \ref{netsym}. Finally, some concluding remarks are presented in Section \ref{conclu}.

\section{An Omnidirectional Relay Scheme}
\label{scheme}

Consider a wireless network of $n$ nodes $\NN=\{1,2,\ldots,n\}$.

Consider the following AWGN wireless network
channel model:
\begin{equation}\label{netcha}
Y_j(t)=\sum_{\stackrel{i\in \NN}{i\neq
j}}g_{i,j}X_i(t)+Z_j(t),\quad \quad \forall\, j\in \NN,\quad
t=1,2,\ldots
\end{equation}
where, $X_i(t)\in \C^1$ and $Y_i(t)\in\C^1$ respectively denote
the signals sent and received by Node $i\in\NN$ at time $t$; $\{ g_{i,j} \in
\C^1 : i \neq j \}$ denote the signal attenuation gains; and
$Z_i(t)$ is zero-mean complex Gaussian noise with variance $N$. Note that we are considering a full-duplex model, i.e., nodes can transmit and receive signals at the same time. However, it will be clear that the main results of this paper can be easily extended to half-duplex models.

Consider the networking problem where each node $i\in\NN$ wants to send the same information at rate $R_i$ (can be zero) to all the nodes in a subset $\TT_i\subset\NN$. Or reversely, each node $i\in\NN$ wants to receive the information sent by all the nodes in some subset $\SS_i\subset\NN$. To achieve this, we design an omnidirectional relay scheme as the following.

We choose a sequence of decode-sets and encode-sets for each node in $\NN$ in the following order. First, for each node $i\in \NN$, choose a subset of $\NN\backslash \{i\}$ as its 1-hop decode-set $\DD_{i(1)}$, and then choose a subset of $\DD_{i(1)}$ as its 1-hop encode-set $\EE_{i(1)}$. That is,
$$
\EE_{i(1)}\subseteq \DD_{i(1)}\subseteq \NN\backslash \{i\}.
$$
Then, for each node $i\in \NN$, choose its 2-hop decode-set and encode-set as
\begin{eqnarray*}
\DD_{i(2)}&\subseteq &\NN\backslash \left\{\{i\}\cup\DD_{i(1)}\right\}\\
\EE_{i(2)}&\subseteq & \left\{\DD_{i(1)}\cup\DD_{i(2)}\right\}\backslash \EE_{i(1)}
\end{eqnarray*}
Sequentially, for $k=3,4,\ldots,L$, where $L$ is some selected finite integer, node $i$'s $k$-hop decode-set and encode-set are chosen as
\begin{eqnarray*}
\DD_{i(k)}&\subseteq &\NN\backslash \left\{\{i\}\cup\DD_{i(1)}\cup \cdots \cup \DD_{i(k-1)}\right\} \\
\EE_{i(k)}&\subseteq & \left\{\DD_{i(1)}\cup\cdots\cup\DD_{i(k)}\right\}\backslash \left\{\EE_{i(1)}\cup\cdots\cup \EE_{i(k-1)}\right\}
\end{eqnarray*}

We use block Markov coding. Consider $B$ blocks of equal length, and in each block $b=1,2,\ldots,B$, denote the message of node $i$ by $w_i(b)$, which is encoded at rate $R_i$.

In block 1, each node $i$ transmits
its own message $w_i(1)$. At the end of block 1, each node $i$
decodes the messages sent by the nodes of its 1-hop decode-set, i.e.,
$\{w_j(1):j\in\DD_{i(1)}\}$.

In block 2,
each node $i$ transmits $\{w_i(2),w_{\EE_{i(1)}}(1)\}$ using the
binning technique, where, $w_{\EE_{i(1)}}(1)$ stands
for $\{w_j(1):j\in \EE_{i(1)}\}$. That is, besides its own message $w_i(2)$, node $i$ also helps transmitting the previous-block messages of the nodes in its 1-hop encode-set, which have been decoded by node $i$ since $\EE_{i(1)}\subseteq \DD_{i(1)}$.
At the end of block 2, each node
$i$ decodes the block-2 messages of the nodes in its 1-hop decode-set and the block-1 messages of the nodes in its 2-hop decode-set, i.e.,
$\{w_{\DD_{i(1)}}(2),w_{\DD_{i(2)}}(1)\}$.

Sequentially, in block $b=3,4,\ldots,$ each node $i$
transmits $\{w_i(b),w_{\EE_{i(1)}}(b-1),\ldots,w_{\EE_{i(b-1)}}(1)\}$
using the binning technique, and decodes
$\{w_{\DD_{i(1)}}(b),\ldots,w_{\DD_{i(b)}}(1)\}$ at the end of block
$b$, where, let $\EE_{i(b)}=\DD_{i(b)}=\emptyset$ when $b>L$, and always set $w_{\emptyset}(l)=\emptyset$ for any $l\geq 1$.

To implement the above omnidirectional relay scheme, we can use
regular encoding/sliding-window decoding with random binning at
each node, as has been used in several simple networks in
\cite{xie07}. Note that random binning can be
replaced by deterministic binning that is easier to implement,
although random binning is simpler to describe in the
achievability proof.

In order to successfully carry out the above omnidirectional relay scheme, obviously, the necessary and sufficient condition is that at the end of each block $b=1,2,3,\ldots$, every node $i\in\NN$ can successfully decode $\{w_{\DD_{i(1)}}(b),\ldots,w_{\DD_{i(b)}}(1)\}$. This is essentially a multi-block multiple-access problem, which will be discussed in detail in Section \ref{lem}.

Apparently, the result of successfully carrying out the omnidirectional relay scheme for $B$ blocks, with $B\gg L$ such that $(B-L)/B\approx 1$, is that each node $i$ receives the messages generated by all the nodes in the set $\bigcup_{k=1}^L\DD_{i(k)}$,  approximately at their initial rates. Therefore, the original networking problem is solved as long as $\SS_i\subseteq \bigcup_{k=1}^L\DD_{i(k)}$ for all $i\in \NN$.

Hence, the key step in the design of the omnidirectional relay scheme is the selection of appropriate decode-sets and encode-sets. The sizes of decode-sets are restricted by the decoding requirement, but should be large enough to finally cover all the intended source nodes. Larger encode-sets result in more messages being helped, but may increase the decoding burden to some nodes that may not be interested in all the messages. It is instructive to note that finally, for any node $i$, the signals transmitted by all the nodes in $\bigcup_{k=1}^L\DD_{i(k)}$ are decoded, either as useful messages, or as useless messages but not causing interference, while the signals transmitted by all the nodes in $\NN\backslash \bigcup_{k=1}^L\DD_{i(k)}$ are not decoded, thus causing interference.

\section{The All-source all-cast problem}
\label{asac}

In order to demonstrate the benefit of the omnidirectional relay scheme,
in this paper, we focus on the special networking problem where all the nodes are independent sources and each node wants to send its information to all the other nodes in the network. That is, we consider the special case where $\TT_i=\NN\backslash\{i\}$ for all $i\in \NN$, or equivalently, $\SS_i=\NN\backslash\{i\}$ for all $i\in \NN$.
Naturally, this can be named as the all-source all-cast problem. To simplify the studies, we only address the case where all rates $R_i$ are equal to some common rate $R$.

We make a very general assumption on the signal attenuation. We only assume that longer distance, higher attenuation. That is, there is a non-increasing function to relate the magnitude of the gains in \dref{netcha} to the distance:
\begin{equation} \label{gainmodel}
|g_{i,j}|=g(d_{i,j}),
\end{equation}
where $d_{i,j}$ is the distance between node $i$ and node $j$, and
$g(\cdot)$ is some non-increasing function. For simplicity, we assume the same transmit power constraint $P$ for all the nodes. Therefore, when a node $i$ is transmitting at its full power, the corresponding received power at another node $j$ is $|g_{i,j}|^2P$.

We will show that for the all-source all-cast problem, it is possible to completely eliminate interference in arbitrarily large wireless networks, and each node can make use of the signals transmitted by all the other nodes.
More importantly, we will show the achievability of the following
common rate for the all-source all-cast problem for some network topologies by the
omnidirectional relay scheme:
\begin{equation}\label{achievable}
R<\frac{1}{n-1}\log\left(1+\frac{{\displaystyle\min_{j}}\sum_{i\neq j
}|g_{i,j}|^2P}{N}\right).
\end{equation}

Obviously, $\sum_{i\neq j}|g_{i,j}|^2P$ is the total
received power at node $j$ if the signals transmitted by different nodes didn't add up coherently at the receiver. This will be the case if independent codebooks are used at different nodes. Then,
${\min_{j}}\sum_{i\neq j}|g_{i,j}|^2P$ corresponds to the node
whose total received power is the least. Since every node needs to
decode all the other $n-1$ sources, \dref{achievable} clearly is the highest
common rate $R$ achievable for the
all-source all-cast problem according to the Shannon formula.

It may be possible to achieve higher rates than \dref{achievable}
by using correlated codebooks at different nodes to boost the received power at some nodes, say, by beamforming or coherent transmission. A method is by using superposition coding as mentioned in the Introduction. However, this may be hard to implement in practice due to, e.g., the lack of channel state information at the transmitters. Moreover, note that
cooperating signals must represent the same information in order to
cooperate, which means that they cannot help the transmission of
other different messages. This may not be a good choice for the
all-source all-cast problem, where the messages to be transmitted
by any two nodes are not completely the same.

We will show that the rate \dref{achievable} is achievable for networks with some kind of symmetry, which include the network depicted in Fig. \ref{onereg} where the nodes are evenly spaced. In the following, we first develop a special version of the omnidirectional relay scheme for the all-source all-cast problem, where network topology is taken into consideration.

\begin{figure}[hbt]
\centering
\includegraphics[width=2.8in]{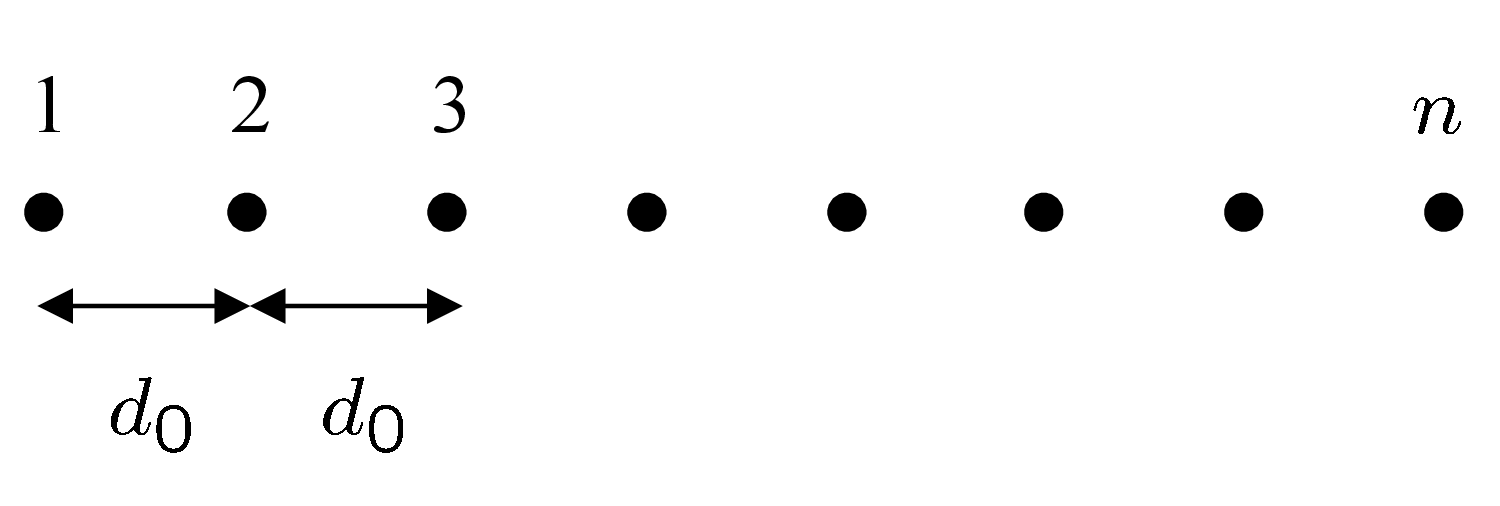}
\caption{A regular network.}
 \label{onereg}
\end{figure}

\subsection{A distance-regulated omnidirectional relay scheme}

We introduce the concept of $k$-hop neighbors in the network in
the following way. First, for each node $i$, define a set of nodes
in its neighborhood as its 1-hop neighbors, and denote the set as
$\NN_{i(1)}$. The way of defining 1-hop neighbors depends on the
network topology and will be specified later on for different
networks. If node $j$ is a 1-hop neighbor of node $i$, it is said
that $j$ can reach $i$ in one hop. If furthermore, $i$ is a 1-hop
neighbor of node $l$, then it is said that $j$ can reach $l$ in
two hops. Similarly, it can be said that a node can reach another
node in $k$ hops, for any positive integer $k$. Now, for each node
$i$, its $k$-hop neighbors is defined as the set of nodes that
can reach it in $k$ hops, but not in any less hops, and denote this
set as $\NN_{i(k)}$. Mathematically, $\NN_{i(k)}$ can be sequentially
defined as
\begin{eqnarray}
\NN_{i(k)}=\{j:j\in \NN_{l(1)} \mbox{ for some }l\in\NN_{i(k-1)}, \quad \label{Nkdef} \\
\mbox{ and } j\notin\{i\}\cup\NN_{i(1)}\cup \cdots\cup\NN_{i(k-1)}\}. \nonumber
\end{eqnarray}
It is clear that for any network of a finite number of nodes, there is a finite number $L_i$ for each $i\in\NN$, such that $\NN_{i(k)}=\emptyset$ for $k>L_i$.

We use block Markov coding. In block 1, each node $i$ transmits
its own message $w_i(1)$. At the end of block 1, each node $i$
decodes at least the messages sent by its 1-hop neighbors
$\{w_j(1):j\in\NN_{i(1)}\}$ (Maybe more can be decoded). In block 2,
each node $i$ transmits $\{w_i(2),w_{\NN_{i(1)}}(1)\}$ using the
binning technique, where for simplicity, $w_{\NN_{i(1)}}(1)$ stands
for $\{w_j(1):j\in \NN_{i(1)}\}$. At the end of block 2, each node
$i$ decodes at least the block-2 messages of its 1-hop neighbors and the block-1 messages of its 2-hop neighbors, i.e.,
$\{w_{\NN_{i(1)}}(2),w_{\NN_{i(2)}}(1)\}$. In block 3, each node $i$
transmits $\{w_i(3),w_{\NN_{i(1)}}(2),w_{\NN_{i(2)}}(1)\}$ using the
binning technique. Generally, in block $b$, each node $i$
transmits $\{w_i(b),w_{\NN_{i(1)}}(b-1),\ldots,w_{\NN_{i(b-1)}}(1)\}$
using the binning technique, and decodes at least
$\{w_{\NN_{i(1)}}(b),\ldots,w_{\NN_{i(b)}}(1)\}$ at the end of block
$b$, where, when the block number is large enough such that
$\NN_{i(b)}=\emptyset$, $w_{\emptyset}(l)=\emptyset$ for any $l\geq 1$.

Obviously, $\EE_{i(k)}$ corresponds to $\NN_{i(k)}$ in this special version, while $\DD_{i(k)}$ can be arbitrary as long as
$$
\EE_{i(1)}\cup\cdots\cup \EE_{i(k)}\subseteq  \DD_{i(1)}\cup\cdots\cup\DD_{i(k)},\quad \mbox{ for any } i\in \NN \mbox{ and } k\geq 1.
$$
In order to solve the all-source all-cast problem where each node needs to decode the messages of all the other nodes, for the networks to be discussed in Section \ref{netsym}, we will choose the 1-hop neighbor sets $\{\NN_{i(1)}:i\in\NN\}$ in a way such that for any $i\in\NN$,
\begin{equation}
\label{cover}
\bigcup_{k=1}^{L_i}\NN_{i(k)}=\NN\backslash\{i\}.
\end{equation}

To show that this scheme works for some networks, we start with a key technical lemma in next section, which discusses a multiple-access decoding based on multiple blocks.

\section{Key Technical Lemma: Multi-block Multiple-Access}
\label{lem}

Consider an AWGN multiple access channel
\begin{equation}
\label{multiacc}
Y(t)=\sum_{i\in \MM} X_i(t)+Z(t),
\end{equation}
where, $\MM=\{1,2,\ldots,m\}$ denotes the set of sources.

According to the well known multiple-access capacity region \cite[Ch.14]{covtho91}, a rate vector $(R_1,\ldots,R_m)$ is achievable if and only if the inequality
\begin{equation}
\label{k1}
\sum_{i\in {\cal S}}R_i<\log\left(1+\frac{\sum_{i\in {\cal S}}P_i}{N}\right)
\end{equation}
holds for all non-empty subsets ${\cal S}\subseteq {\cal M}$. Namely, if each source $i\in \MM$ encodes its message $w_i$ at rate $R_i$ with independent Gaussian block codewords $\b{X}_i(w_i)$ with power $P_i$, then \dref{k1} is the necessary and sufficient condition such that $\{w_1,w_2,\ldots,w_m\}$ can be decoded, in the sense that the decoding error can be made arbitrarily small by increasing the block length.

Obviously, \dref{k1} needs to hold for all nonempty $\SS\subseteq \MM$ in order to decode $\{w_1,w_2,\ldots,w_m\}$. However, it may not be so commonly recognized that as long as \dref{k1} holds for the one $\SS = \MM$, there must be some nonempty subset of $\{w_1,w_2,\ldots,w_m\}$ that can be decoded. This is formally stated as the following lemma.
\begin{lemma}
\label{lemma1}
For the multiple access channel \dref{multiacc}, with each source $i\in \MM$ sending a message $w_i$ at rate $R_i$ with power $P_i$, there always exists some nonempty subset of $\{w_1,w_2,\ldots,w_m\}$ that can be decoded, as long as the following inequality holds:
\begin{equation}
\label{k2}
\sum_{i\in {\cal M}}R_i<\log\left(1+\frac{\sum_{i\in {\cal M}}P_i}{N}\right)
\end{equation}
i.e., \dref{k1} with $\SS=\MM$.
\end{lemma}

\begin{proof}
We use a contradiction argument. Suppose \dref{k1} doesn't hold for some $\A\subset \MM$, i.e.,
\begin{equation}
\label{k3}
\sum_{i\in {\A}}R_i\geq\log\left(1+\frac{\sum_{i\in {\A}}P_i}{N}\right).
\end{equation}
Then taking the difference between \dref{k2} and \dref{k3}, we have
\begin{equation}
\label{k4}
\sum_{i\in \A^c}R_i<\log\left(1+\frac{\sum_{i\in \A^c}P_i}{N_{\A}}\right)
\end{equation}
where, $\A^c=\MM\backslash\A$, and $N_{\A}=\sum_{i\in \A}P_i+ N$. Now, by comparing \dref{k4} with \dref{k2}, we arrive at the same situation as \dref{k2} with $\MM$ replaced by $\A^c$, and $N$ replaced by $N_\A$. Similarly, if the inequality
\begin{equation}
\label{k5}
\sum_{i\in {\cal S}}R_i<\log\left(1+\frac{\sum_{i\in {\cal S}}P_i}{N_\A}\right)
\end{equation}
holds for all nonempty $\SS\subseteq \A^c$, then the subset of messages $\{w_i: \,i\in\A^c\}$ can be decoded; Otherwise, if \dref{k5} doesn't hold for some $\BB\subset\A^c$, the process can be continued with $\BB^c=\A^c\backslash \BB$. As the size of the subset decreases, we must be able to reach a nonempty subset where all the necessary inequalities of the type \dref{k5} hold, and thus the messages can be decoded. This is obvious, since if the process continues without stopping, it must reach a subset with only one source, and by then, the single inequality like \dref{k4} suffices for the decoding.

Therefore, we proved that if \dref{k2} holds, there must exist a nonempty subset $\MM_2\subseteq \MM$ such that $\{w_i:i\in \MM_2\}$ can be decoded, while $\{w_i:i\in \MM_1\}$ with $\MM_1=\MM\backslash \MM_2$ cannot.
\end{proof}

Now, in our block Markov coding setting with relays, the nodes help each other to transfer messages. To put into this perspective, let us consider a two-block decoding situation where in the first block $\{w_i(1):i\in \MM_2\}$ are decoded while $\{w_i(1):i\in \MM_1\}$ are not, and in the second block, each node $i\in\MM_2$ helps transmitting some messages from $\{w_i(1):i\in \MM_1\}$ besides its own message $w_i(2)$. The goal now is to decode $\{w_i(2):i\in \MM_2\}\cup\{w_i(1):i\in \MM_1\}$ at the end of the second block. In consistency with our notation earlier, denote $w_{\MM_1}(1)=\{w_i(1):i\in \MM_1\}$, $w_{\MM_2}(2)=\{w_i(2):i\in \MM_2\}$, and  $\{w_{\MM_2}(2),w_{\MM_1}(1)\}=\{w_i(2):i\in \MM_2\}\cup\{w_i(1):i\in \MM_1\}$.

Denote $\JJ_i\subset \MM$ as the set of nodes that node $i$ helps in the second block, i.e., node $i$ sends a codeword $\b{X}_i(w_i(2),w_{\JJ_i}(1))$ by binning the multiple messages in the second block. Reversely, denote $\II_i\subset \MM$ as the set of nodes that will help node $i$ to transmit $w_i(1)$ in the second block.

For any subset $\SS\subseteq \MM$, let $\SS_1=\SS\cap \MM_1$, and let
\begin{equation}
\label{ds}
\SS_2=(\SS\cap \MM_2)\cup(\bigcup_{i\in\SS_1}\II_i\cap\MM_2).
\end{equation}
That is, $\SS_2$ also consists of nodes from $\MM_2$ that may not be in $\SS$, but are helping transmitting $w_{\SS_1}(1)$.
Then, it can be easily verified with a typical sequence argument that $\{w_{\MM_2}(2),w_{\MM_1}(1)\}$ can be decoded if and only if for any nonempty subset $\SS\subseteq \MM$,
\begin{equation}
\label{k6}
\sum_{i\in {\cal S}}R_i<\log\left(1+\frac{\sum_{i\in {\SS_1}}P_i}{N}\right) +\log\left(1+\frac{\sum_{i\in {\SS_2}}P_i}{\sum_{i\in {\MM_1}}P_i+N}\right)
\end{equation}
where the first term is the contribution of the nodes in $\SS_1$ from the first block, and the second term is the contribution of the nodes in $\SS_2$ from the second block. Actually, it is rather instructive to think of the constraints \dref{k6} for all nonempty $\SS\subseteq \MM$ as a two-block multiple-access region. 

Although it is necessary that the inequality \dref{k6} should hold for all nonempty $\SS\subseteq \MM$ in order to decode $\{w_{\MM_2}(2),w_{\MM_1}(1)\}$, as in the case of one-block multiple-access discussed earlier, we will show that the following single inequality
\begin{equation}
\label{k7}
\sum_{i\in {\MM}}R_i<\log\left(1+\frac{\sum_{i\in {\MM_1}}P_i}{N}\right) +\log\left(1+\frac{\sum_{i\in {\MM_2}}P_i}{\sum_{i\in {\MM_1}}P_i+N}\right)
\end{equation}
i.e., \dref{k6} with $\SS= \MM$, is enough to ensure that some nonempty subset of $\{w_{\MM_2}(2),w_{\MM_1}(1)\}$ can be decoded.

We still use a contradiction argument. If \dref{k6} holds for all nonempty $\SS\subseteq \MM$, then $\{w_{\MM_2}(2),w_{\MM_1}(1)\}$ can be decoded; Otherwise, if for some nonempty $\A\subset \MM$, \dref{k6} doesn't hold, i.e.,
\begin{equation}
\label{k8}
\sum_{i\in {\A}}R_i\geq\log\left(1+\frac{\sum_{i\in {\A_1}}P_i}{N}\right) +\log\left(1+\frac{\sum_{i\in {\A_2}}P_i}{\sum_{i\in {\MM_1}}P_i+N}\right)
\end{equation}
then taking the difference between \dref{k7} and \dref{k8}, we have
\begin{equation}
\label{k9}
\sum_{i\in \A^c}R_i<\log\left(1+\frac{\sum_{i\in {\A_1^c}}P_i}{\sum_{i\in {\A_1}}P_i+ N}\right) +\log\left(1+\frac{\sum_{i\in {\A_2^c}}P_i}{\sum_{i\in {\A_2}}P_i+ \sum_{i\in {\MM_1}}P_i +N}\right)
\end{equation}
where, ${\A^c}=\MM\backslash \A$, ${\A_1^c}=\MM_1\backslash \A_1$, and ${\A_2^c}=\MM_2\backslash \A_2$. By the definition \dref{ds}, it simply follows that $\A\subseteq \A_1\cup\A_2$ and $\A^c\supseteq \A_1^c\cup\A_2^c$. Hence, by replacing $\A^c$ with $\A_1^c\cup\A_2^c$ in the left-hand-side of \dref{k9}, we have
\begin{equation}
\label{k10}
\sum_{i\in \A_1^c\cup\A_2^c}R_i<\log\left(1+\frac{\sum_{i\in {\A_1^c}}P_i}{\sum_{i\in {\A_1}}P_i+ N}\right) +\log\left(1+\frac{\sum_{i\in {\A_2^c}}P_i}{\sum_{i\in {\A_2}}P_i+ \sum_{i\in {\MM_1}}P_i +N}\right).
\end{equation}
This is the same situation as \dref{k7} with $\MM$ replaced by $\A_1^c\cup\A_2^c$, $\MM_1$ replaced by $\A_1^c$, $\MM_2$ replaced by $\A_2^c$, and some adjustment of the noises. Now, the messages to be decoded are $\{w_{\A_2^c}(2),w_{\A_1^c}(1)\}$. As in the case of one-block multiple-access discussed earlier, such a process can be continued until we find a nonempty subset of $\{w_{\MM_2}(2),w_{\MM_1}(1)\}$ that can be decoded.

Therefore, we proved that the inequality \dref{k7} alone ensures that there always exists a nonempty subset of $\{w_{\MM_2}(2),w_{\MM_1}(1)\}$ that can be decoded. Note that by combining the two terms on the right-hand-side, \dref{k7} becomes
\begin{equation}
\label{k11}
\sum_{i\in {\MM}}R_i<\log\left(1+\frac{\sum_{i\in {\MM}}P_i}{N}\right)
\end{equation}
which is exactly the same as \dref{k2}. In other words, the inequality \dref{k2} or \dref{k11} makes sure that there are always some messages that can be decoded, no matter whether it is one-block multiple-access, or two-block multiple access with relays.

It is now clear that generally we have the following conclusion for $K$-block multiple-access with relays.

\begin{lemma}
\label{keylemma}
Consider a $K$-block decoding situation where $\{w_{\MM_K}(K),\ldots,w_{\MM_1}(1)\}$ are to be decoded for some disjoint subsets $\MM_k,\,k=1,\ldots,K$ with $\bigcup_{k=1}^K\MM_k=\MM$, or equivalently to say, that $\{w_{\MM_k}(k-1),\ldots,w_{\MM_{k}}(1)\}$ have been decoded for $k=2,\ldots,K$. During each block $k=2,\ldots,K$, every node $i\in \MM_k$ helps transmitting a subset of $\{w_{\MM_{k-1}}(k-1),\ldots,w_{\MM_1}(1)\}$ besides its own message $w_i(k)$ with the binning technique. Then there is always a nonempty subset of $\{w_{\MM_K}(K),\ldots,w_{\MM_1}(1)\}$ that can be decoded if \dref{k11} holds.
\end{lemma}

\section{Networks with Symmetric Traffic}
\label{netsym}

After the discussion of last section, it is clear that no matter how complicated the relay situation is, at the end of each block $b$, every node $i$ can always decode the new messages $w(b)$ of a nonempty set of nodes $\GG_i(b)\subset \NN$, under the condition \dref{achievable}. (More detailed arguments about this will be presented in the proof of Theorem \ref{th1}). Then in order to successfully carry out the distance-regulated omnidirectional relay scheme presented in Section \ref{asac}, we only need to make sure that for each $i\in\NN$
\begin{equation}
\label{ms}
\NN_{i(1)}\subseteq \GG_i(b), \quad \mbox{ for all }b=1,\ldots,B.
\end{equation}

Due to the monotonicity of the power attenuation model \dref{gainmodel}, messages sent by nodes that are closer are generally easier to decode, and therefore it is natural to choose $\NN_{i(1)}$ as a set composed of the closest nodes. In view of the requirement \dref{ms}, it is preferable to put as few as possible nodes into $\NN_{i(1)}$. However, for the all-source all-cast problem, each $\NN_{i(1)}$ should contain sufficiently many nodes so that \dref{cover} holds, i.e., the whole network will be covered and each node will decode the messages of all the other nodes.

When we are sure that there are some nodes whose messages can be decoded but not knowing how many of them there are, it is not clear whether the messages of all the nodes in $\NN_{i(1)}$ can be decoded if $\NN_{i(1)}$ contains more than one nodes. However, there is a special situation where we can be sure, i.e., when there is some kind of symmetry to all the nodes in $\NN_{i(1)}$, in the sense that if one of them can be decoded, the others certainly can. Two simple examples of this are shown in Fig. \ref{symmetry}, where clearly, for any node $i$, the traffic is symmetric on both sides. For each node $i$, by choosing $\NN_{i(1)}$ as its two neighboring nodes, which are mostly easy to decode, it is certain that both of them can be decoded. Since \dref{cover} obviously holds by this definition, the all-source all-cast problem for these networks is solved under the condition \dref{achievable}.

\begin{figure}[hbt]
\centering
\includegraphics[width=4.5in]{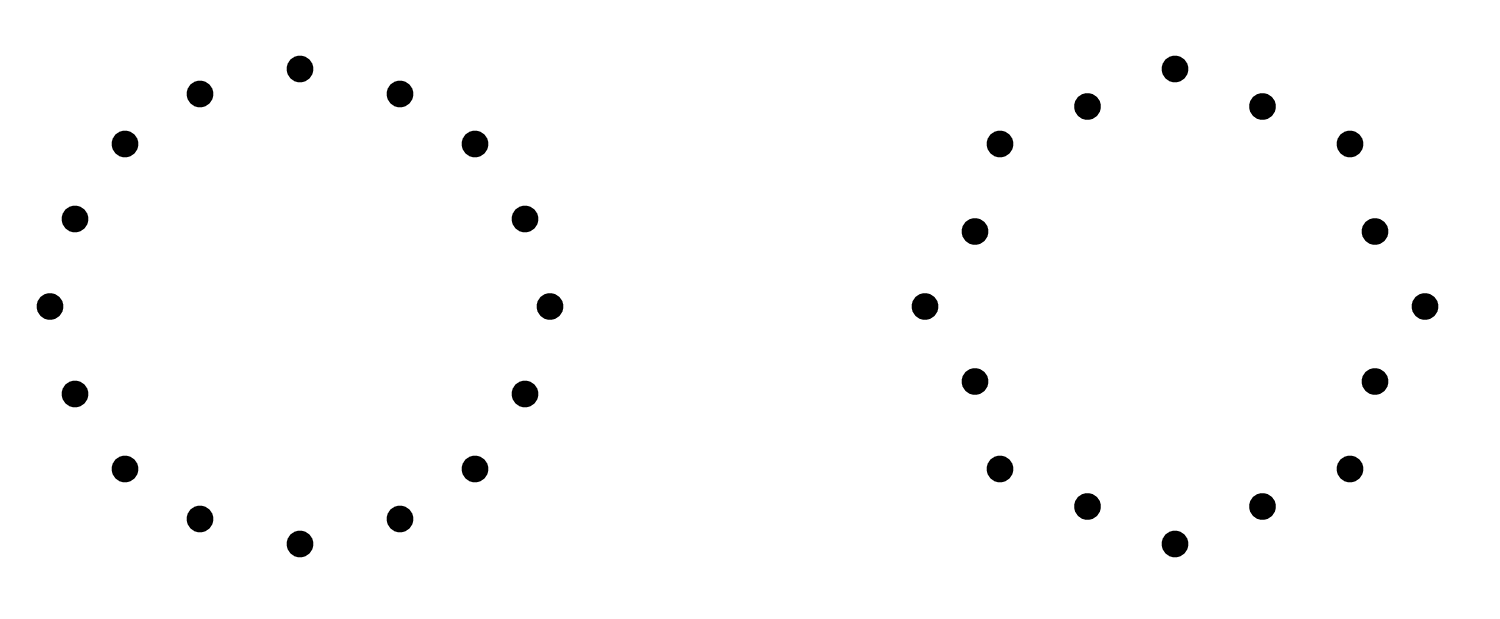}
\caption{Two symmetric networks.}
 \label{symmetry}
\end{figure}

The network depicted in Fig. 1 is not completely symmetric. For any node not in the center, i.e., $i\neq (n+1)/2$, the traffic on one side is heavier than the other side. However, there is still some kind of symmetry as we will show later on, so that any non-boundary node $i\notin \{1,n\}$ can decode the messages of both its neighbors $\{i-1,i+1\}$ simultaneously, under the condition \dref{achievable}.

We will first prove this for more general network topologies, and then the regular topology in Fig. \ref{onereg} will follow as a simple corollary. Alternatively, a direct proof for the regular network in Fig. 1 has been presented in \cite{xie08a}.

\begin{figure}[hbt]
\centering
\includegraphics[width=2.8in]{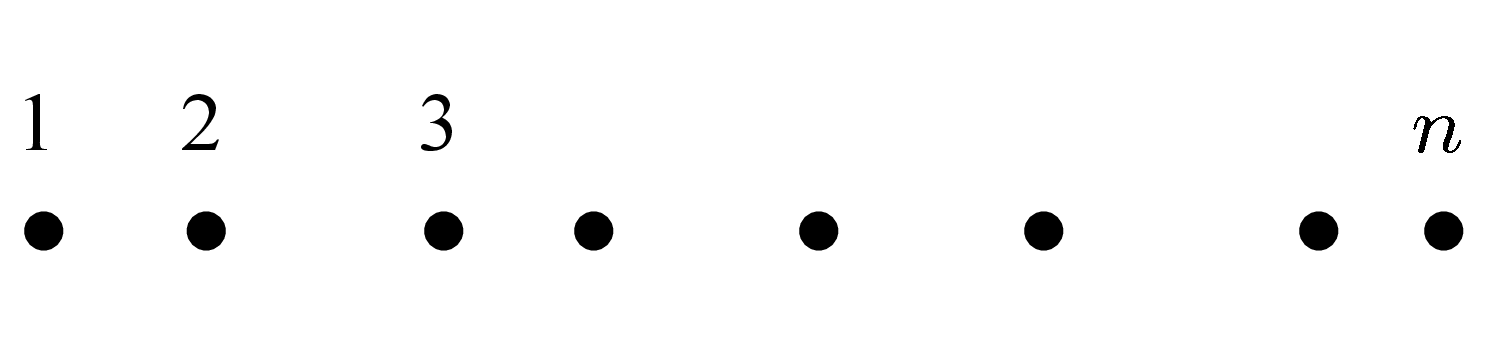}
\caption{A general one-dimensional network.}
\label{onedim}
\end{figure}

Consider a general wireless network of $n$ nodes located on a straight
line, labeled sequentially by $1,2,\ldots,n$, as depicted in Fig. \ref{onedim}. It is convenient to introduce the notation
$$
P_{i,j}=|g_{i,j}|^2 P
$$
for any $i,j\in \{1,2,\ldots,n\}$.

For any node $i\notin\{1,n\}$, let its 1-hop neighbors be
$\NN_{i(1)}=\{i-1,i+1\}$. Let $\NN_{1(1)}=\{2\}$ and
$\NN_{n(1)}=\{n-1\}$. We will show that the distance-regulated omnidirectional
relay scheme presented in Section \ref{asac} works for this one-dimensional network as long as the common rate $R$ satisfies \dref{achievable} and the following two symmetric sets of constraints for every $i\in \{2,3,\ldots,n-1\}$:
\begin{description}
\item[i)]
For any $\ell=1,\ldots,i-2$, at least one of the following two inequalities holds:
\begin{eqnarray}
(\ell+n-i)R&<&\log \left(1+\frac{P_{1,i}+\cdots+P_{\ell,i}+P_{i+1,i}+\cdots+P_{n,i}}{N} \right) \label{left1}  \\ \mbox{or} \quad \quad\quad
(n-i)R&<&\log \left(1+\frac{P_{i+1,i}+\cdots+P_{n,i}}{P_{1,i}+\cdots+P_{\ell,i}+N} \right) \label{left2}
\end{eqnarray}
\item[ii)]
For any $r=i+2,\ldots,n$,
at least one of the following two inequalities holds:
\begin{eqnarray}
(i+n-r)R&<&\log \left(1+\frac{P_{1,i}+\cdots+P_{{i-1},i}+P_{r,i}+\cdots+P_{n,i}}{N} \right) \label{right1}  \\ \mbox{or} \quad \quad\quad
(i-1)R&<&\log \left(1+\frac{P_{1,i}+\cdots+P_{{i-1},i}}{P_{r,i}+\cdots+P_{n,i}+N} \right) \label{right2}
\end{eqnarray}
\end{description}
In other words, we have the following theorem.

\begin{theorem}
\label{th1}
For the one-dimensional wireless network, a common rate $R$ is
achievable for the all-source all-cast problem with the
omnidirectional relay scheme, if it satisfies \dref{achievable}, and also for
every $i\in \{2,3,\ldots,n-1\}$, the above constraints i) and ii) hold.
\end{theorem}

\begin{proof}
With the distance-regulated omnidirectional relay scheme, we only need to show that at the end of each block $b$, every node $i$ can decode $w_{\NN_{i(1)}}(b)$.

Obviously, for any node $i$ and at the end of any block $b$, there are a sequence of disjoint subsets  $\MM_k$, $k=1,\ldots,b$ (can be empty) with $\bigcup_{k=1}^b\MM_k=\NN\backslash\{i\}$ such that $\{w_{\MM_b}(b),w_{\MM_{b-1}}(b-1),\ldots,w_{\MM_1}(1)\}$ are to be decoded, or equivalently to say, that $\{w_{\MM_k}(k-1),\ldots,w_{\MM_{k}}(1)\}$ have been decoded for any $k=2,\ldots,b$ in the previous blocks. Then according to Lemma \ref{keylemma}, there must exist a nonempty subset of $\{w_{\MM_b}(b),w_{\MM_{b-1}}(b-1),\ldots,w_{\MM_1}(1)\}$ that can be decoded, due to
\begin{equation}
\label{on1}
(n-1)R<\log\left(1+\frac{P_{1,i}+\cdots+P_{{i-1},i}+P_{i+1,i}+\cdots+P_{n,i}}{N} \right)
\end{equation}
which follows from \dref{achievable}. If this nonempty subset is disjoint with $w_{\MM_b}(b)$, then after the decoding, we arrive at a similar situation with another sequence of disjoint $\MM'_k$, $k=1,\ldots,b$ with $\bigcup_{k=1}^b\MM_k=\NN\backslash\{i\}$. Then Lemma \ref{keylemma} can be applied again with \dref{on1} so that more messages can be decoded. This process can be continued as long as all nodes in $\NN\backslash \{i\}$ have messages to be decoded. In other words, finally, there must be a nonempty subset $\MM_b^*\subseteq \NN\backslash \{i\}$ such that $w_{\MM_b^*}(b)$ can be decoded at the end of block $b$.

According to the relay structure and the monotonicity of the power attenuation, $\MM_b^*$ can only be one of the following three types of subsets of nodes:
$\{\ell,\ldots,i-1\}$ for some $\ell<i$;
$\{i+1,\ldots,r\}$ for some $r>i$; or $\{\ell,\ldots,i-1,i+1,\ldots,r\}$ for some $\ell<i<r$. This is simply based on the observation that on either side, it is always easier to decode messages from nodes closer. If $i$ is a boundary node, i.e., $1$ or $n$, then only one of the first two types is possible and clearly $\NN_{i(1)}\subset\MM_b^*$. Now, for a non-boundary node $i\in\{2,3,\ldots,n-1\}$, all three types are possible. If $\MM_b^*$ is of the third type, then clearly, $\NN_{i(1)}\subset\MM_b^*$ and the proof is finished. If $\MM_b^*$ is of the first type, then Lemma \ref{keylemma} still can be applied with either \dref{left1} or \dref{left2} continually until $w_{i+1}(b)$ is decoded. Note that the case \dref{left2} is different from \dref{left1} in the sense that there is no intension to decode the messages of the nodes $\{1,\ldots,\ell\}$, and their transmissions are treated as noise.  Actually, they may not be all causing interference in all blocks, and hence, the condition needed to apply Lemma \ref{keylemma} may be weaker than \dref{left2}. Symmetrically, it can be shown that $w_{i-1}(b)$ will be decoded based on either \dref{right1} or \dref{right2} if $\MM_b^*$ is of the second type. Therefore, we've shown that $w_{\NN_{i(1)}}(b)$ will always be decoded. This concludes the proof.
\end{proof}

Now, we show that for the regular network in Fig. \ref{onereg}, any rate satisfying \dref{achievable} must satisfy the constraints i) and ii). Thus, the rate \dref{achievable} is achievable.

Due to the equal separation distance $d_0$ and the power gain model \dref{gainmodel}, it is convenient to define
$$
P_i=g(id_0)P \quad \mbox{ for any } i\geq 1.
$$
Then, $P_{i,j}=P_{|i-j|}$ for any $i\neq j$. According to the monotonicity of the function $g(\cdot)$, we have
\begin{equation}
\label{Porder}
P_1\geq P_2\geq \cdots \geq P_{n-1}.
\end{equation}
With this new notation, \dref{achievable} becomes
\begin{equation}
\label{achiev1}
R<\frac{1}{n-1}\log\left(1+\frac{P_1+P_2+\ldots+P_{n-1}}{N}\right)
\end{equation}
where, the total received power corresponds to any one of the boundary nodes, and is the smallest among all the nodes. The constraints i) and ii) become: For every $i\in \{2,3,\ldots,n-1\}$,
\begin{description}
\item[i)]
For any $\ell=1,\ldots,i-2$, at least one of the following two inequalities holds:
\begin{eqnarray}
(\ell+n-i)R&<&\log\left(1+\frac{\sum_{j=i-\ell}^{i-1}P_{j}+\sum_{j=1}^{n-i}P_{j}}{N} \right) \label{left1'}  \\ \mbox{or} \quad \quad\quad
(n-i)R&<&\log\left(1+\frac{\sum_{j=1}^{n-i}P_{j}}{\sum_{j=i-\ell}^{i-1}P_{j}+N} \right) \label{left2'}
\end{eqnarray}
\item[ii)]
For any $r=i+2,\ldots,n$,
at least one of the following two inequalities holds:
\begin{eqnarray}
(i+n-r)R&<&\log\left(1+\frac{\sum_{j=1}^{i-1}P_{j}+\sum_{j=r-i}^{n-i}P_{j}}{N} \right) \label{right1'}  \\ \mbox{or} \quad \quad\quad
(i-1)R&<&\log\left(1+\frac{\sum_{j=1}^{i-1}P_{j}}{\sum_{j=r-i}^{n-i}P_{j}+N} \right) \label{right2'}
\end{eqnarray}
\end{description}

Now, we verify that at least one of \dref{left1'} and \dref{left2'} must hold. First, note that by the concavity of the logarithmic function, it follows from \dref{Porder} and \dref{achiev1} that for any $1\leq k\leq n-1$,
\begin{equation}
\label{con1}
kR<\log\left(1+\frac{P_1+P_2+\ldots+P_{k}}{N}\right).
\end{equation}
Specially, when $k=i-1$, we have
\begin{equation}
\label{con2}
(i-1)R<\log\left(1+\frac{P_1+P_2+\ldots+P_{i-1}}{N}\right).
\end{equation}
If $i-\ell\leq n-i+1$, by \dref{Porder}, we have
$$
\sum_{j=i-\ell}^{i-1}P_{j}+\sum_{j=1}^{n-i}P_{j}\geq \sum_{j=1}^{n-i+\ell}P_{j}
$$
and thus, by \dref{con1} with $k=\ell+n-i$, \dref{left1'} holds. Otherwise, if $i-\ell> n-i+1$, we check the following inequality
\begin{equation}
\label{ch1}
\ell R<\log\left(1+\frac{\sum_{j=i-\ell}^{i-1}P_{j}}{N} \right).
\end{equation}
If \dref{ch1} holds, then by \dref{con2}, \dref{Porder} and the concavity of the logarithmic function, \dref{left1'} follows. Otherwise, if \dref{ch1} doesn't hold, i.e.,
\begin{equation}
\label{ch3}
\ell R\geq \log\left(1+\frac{\sum_{j=i-\ell}^{i-1}P_{j}}{N}\right),
\end{equation}
taking the difference between \dref{con2} and \dref{ch3}, we have
\begin{equation}
\label{ch4}
(i-1-\ell)R< \log\left(1+\frac{\sum_{j=1}^{i-\ell-1}P_{j}}{\sum_{j=i-\ell}^{i-1}P_{j}+N} \right).
\end{equation}
Then again by \dref{Porder} and the concavity of the logarithmic function, we have \dref{left2'}.

Similarly, by symmetry, we can show that at least one of \dref{right1'} and \dref{right2'} must hold. Therefore, we arrive at the following theorem.

\begin{theorem}
\label{th2}
For the one-dimensional regular wireless network in Fig. \ref{onereg}, the common rate \dref{achiev1} is
achievable for the all-source all-cast problem with the
omnidirectional relay scheme.
\end{theorem}

\begin{figure}[hbt]
\centering
\includegraphics[width=2.8in]{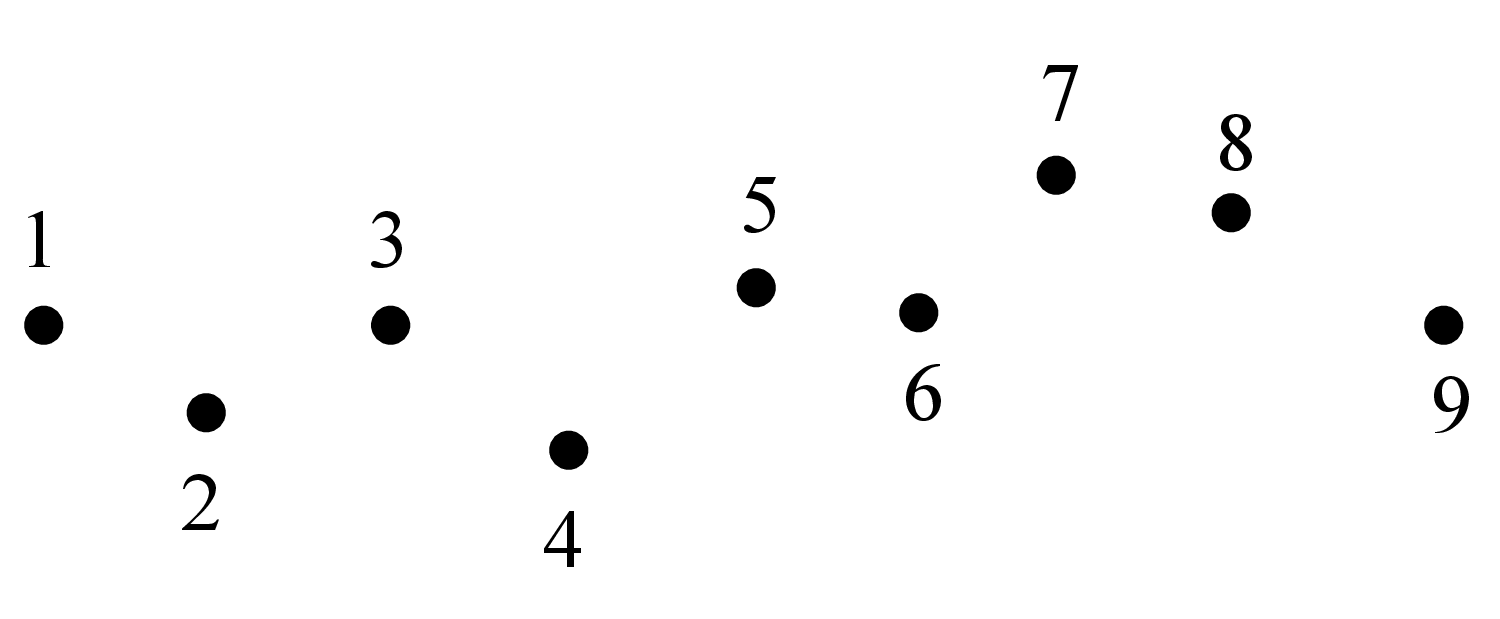}
\caption{A general network with nodes clearly ordered by distance.}
\label{oneorder}
\end{figure}
\begin{figure}[hbt]
\centering
\includegraphics[width=2.8in]{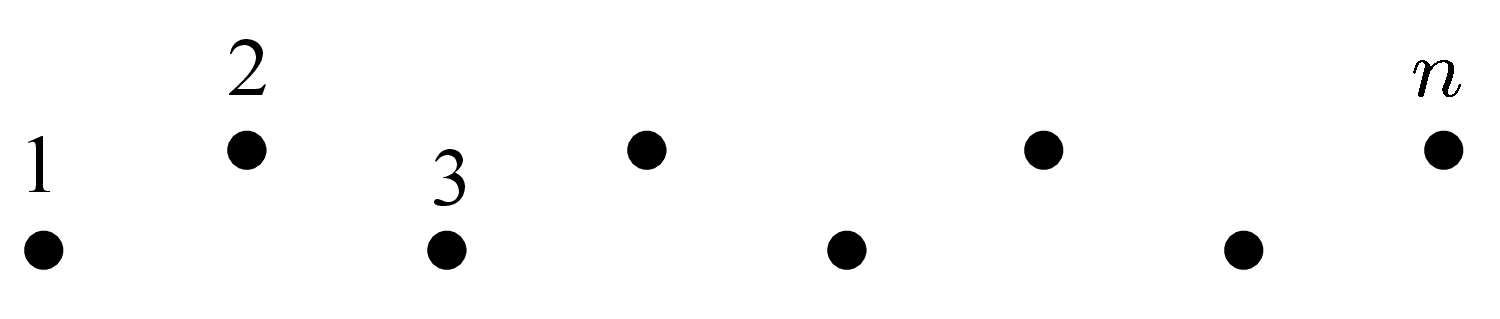}
\caption{A regular network with a clear ordering of nodes by distance.}
\label{oneorderreg}
\end{figure}

\begin{remark}
In all the arguments above, obviously, it is not necessary for all the nodes to be located on a straight line, as long as they can be clearly ordered in terms of the distances, i.e., there is a way of labeling the nodes so that $d_{i,j}\leq d_{i,k}$ for any $i<j<k$, or $k<j<i$. One such example is shown in Fig. \ref{oneorder}, and a regular case is shown in Fig. \ref{oneorderreg}. In such cases, Theorem \ref{th1} or \ref{th2} still applies.
\end{remark}

\section{Conclusion}

We developed an omnidirectional relay scheme for wireless networks with multiple sources, where each node can simultaneously relay multiple messages in different directions by binning them into a single signal. This scheme also exploits the broadcast nature of wireless communication, such that one node helps multiple nodes, and multiple nodes help one node. In the extreme, this scheme is capable of completely eliminating interference in the whole network, and specially, for the all-source all-cast problem where all messages are of interest, each node can benefit from the signals transmitted by all the other nodes. We also demonstrated some kind of optimality of this scheme by showing that it achieves the maximum rate possible for some networks if no beamforming is performed.

We proposed a distance-regulated networking framework, which was shown to work well for some networks. To deal with more general problems, the neighborhoods can be selected not only based on the topology, but also on other factors such as the communication rates, interference, etc. It is also possible to make the omnidirectional relay framework presented in Section \ref{scheme} more general by introducing layered coding structure at each node. This will admit superposition coding for beamforming, and will also make it possible to transmit different messages for different nodes, as in the basic scheme for the broadcast channel \cite{cov72}. Much remains to be done.

\label{conclu}

\bibliographystyle{ieeetr}
\baselineskip=1.0\normalbaselineskip 
\bibliography{bibfile}

\end{document}